\documentclass[11pt]{article}
\usepackage{setspace}
\usepackage{amsfonts}
\usepackage{amsmath}
\usepackage{amssymb}
\usepackage{amsthm}
\usepackage{enumerate}
\usepackage{mathrsfs}
\numberwithin{equation}{section}
\newtheorem{theorem}{Theorem}[section]
\newtheorem{lemma}[theorem]{Lemma}

\newtheorem{corollary}[theorem]{Corollary}
\newtheorem{observation}[theorem]{Observation}
\newtheorem{expl}[theorem]{Example}

\newenvironment{definition}[1][Definition.]{\begin{trivlist}
\item[\hskip \labelsep {\bfseries #1}]}{\end{trivlist}}

\begin{document}
\bibliographystyle{plain}
\title{Mutual Dimension and Random Sequences\footnote{This research was supported in part by National Science Foundation Grants 0652519, 1143830, 124705, and 1545028. Part of the second author's work was done during a sabbatical at Caltech and the Isaac Newton Institute for Mathematical Sciences at the University of Cambridge. A preliminary version of part of this work was presented at the Fortieth International Symposium on Mathematical Foundations of Computer Science, August 24-28, 2015, in Milano, Italy.}}
\author{Adam Case and Jack H. Lutz\\
Department of Computer Science\\
Iowa State University\\
Ames, IA 50011 USA}
\date{}
\maketitle
\begin{abstract}
If $S$ and $T$ are infinite sequences over a finite alphabet, then the {\it lower} and {\it upper mutual dimensions} $mdim(S:T)$ and $Mdim(S:T)$ are the upper and lower densities of the algorithmic information that is shared by $S$ and $T$.  In this paper we investigate the relationships between mutual dimension and {\it coupled randomness}, which is the algorithmic randomness of two sequences $R_1$ and $R_2$ with respect to probability measures that may be dependent on one another.  For a restricted but interesting class of coupled probability measures we prove an explicit formula for the mutual dimensions $mdim(R_1:R_2)$ and $Mdim(R_1:R_2)$, and we show that the condition $Mdim(R_1:R_2) = 0$ is necessary but not sufficient for $R_1$ and $R_2$ to be independently random.

We also identify conditions under which Billingsley generalizations of the mutual dimensions $mdim(S:T)$ and $Mdim(S:T)$ can be meaningfully defined; we show that under these conditions these generalized mutual dimensions have the ``correct" relationships with the Billingsley generalizations of $dim(S)$, $Dim(S)$, $dim(T)$, and $Dim(T)$ that were developed and applied by Lutz and Mayordomo; and we prove a divergence formula for the values of these generalized mutual dimensions.
\end{abstract}
\section{Introduction}
Algorithmic information theory combines tools from the theory of computing and classical Shannon information theory to create new methods for quantifying information in an expanding variety of contexts. Two notable and related strengths of this approach that were evident from the beginning \cite{Kolm65} are its abilities to quantify the information in and to assess the randomness of \emph{individual} data objects.

Some useful mathematical objects, such as real numbers and execution traces of nonterminating processes, are intrinsically infinitary. The randomness of such objects was successfully defined very early \cite{Mart66} but it was only at the turn of the present century \cite{jLutz03a,Lutz03b} that ideas of Hausdorff were reshaped in order to define \emph{effective fractal dimensions}, which quantify the densities of algorithmic information in such infinitary objects. Effective fractal dimensions, of which there are now many, and their relations with randomness are now a significant part of algorithmic information theory \cite{DowHir10}.

Many scientific challenges require us to quantify not only the information in an individual object, but also the information \emph{shared} by two objects. The \emph{mutual information} $I(X;Y)$ of classical Shannon information theory does something along these lines, but for two probability spaces of objects rather than for two individual objects \cite{bCovTho06}. The \emph{algorithmic mutual information} $I(x:y)$, defined in terms of Kolmogorov complexity \cite{bLiVit08}, quantifies the information shared by two individual finite objects $x$ and $y$.

The present authors recently developed the \emph{mutual dimensions} $mdim(x:y)$ and $Mdim(x:y)$ in order to quantify the density of algorithmic information shared by two infinitary objects $x$ and $y$ \cite{CasLut15}. The objects $x$ and $y$ of interest in \cite{CasLut15} are points in Euclidean spaces $\mathbb{R}^n$ and their images under computable functions, so the fine-scale geometry of $\mathbb{R}^n$ plays a major role there.

In this paper we investigate mutual dimensions further, with objectives that are more conventional in algorithmic information theory. Specifically, we focus on the lower and upper mutual dimensions $mdim(S:T)$ and $Mdim(S:T)$ between two sequences $S,T \in \Sigma^{\infty}$, where $\Sigma$ is a finite alphabet. (If $\Sigma = \{0,1\}$, then we write $\bf{C}$ for the \emph{Cantor space} $\Sigma^{\infty}$.) The definitions of these mutual dimensions, which are somewhat simpler in $\Sigma^{\infty}$ than in $\mathbb{R}^n$, are implicit in \cite{CasLut15} and explicit in section 2 below.

Our main objective here is to investigate the relationships between mutual dimension and \emph{coupled randomness}, which is the algorithmic randomness of two sequences $R_1$ and $R_2$ with respect to probability measures that may be dependent on one another. In section 3 below we formulate coupled randomness precisely, and we prove our main theorem, Theorem \ref{random resp}, which gives an explicit formula for $mdim(R_1:R_2)$ and $Mdim(R_1:R_2)$ in a restricted but interesting class of coupled probability measures. This theorem can be regarded as a ``mutual version'' of Theorem 7.7 of \cite{Lutz03b}, which in turn is an algorithmic extension of a classical theorem of Eggleston \cite{Eggl49,Bill65}. We also show in section 3 that $Mdim(R_1:R_2) = 0$ is a necessary, but not sufficient condition for two random sequences $R_1$ and $R_2$ to be independently random.

In 1960 Billingsley investigated generalizations of Hausdorff dimension in which the dimension itself is defined ``through the lens of'' a given probability measure \cite{Bill60,Caja82}. Lutz and Mayordomo developed the effective Billingsley dimensions $dim^{\nu}(S)$ and $Dim^{\nu}(S)$, where $\nu$ is a probability measure on $\Sigma^{\infty}$, and these have been useful in the algorithmic information theory of self-similar fractals \cite{jLutMay08,GLMM14}.

In section 4 we investigate ``Billingsley generalizations'' $mdim^{\nu}(S:T)$ and $Mdim^{\nu}(S:T)$ of $mdim(S:T)$ and $Mdim(S:T)$, where $\nu$ is a probability measure on $\Sigma^{\infty} \times \Sigma^{\infty}$. These turn out to make sense only when $S$ and $T$ are \emph{mutually normalizable}, which means that the normalizations implicit in the fact that these dimensions are \emph{densities} of shared information are the same for $S$ as for $T$. We prove that, when mutual normalizability is satisfied, the Billingsley mutual dimensions $mdim^{\nu}(S:T)$ and $Mdim^{\nu}(S:T)$ are well behaved. We also identify a sufficient condition for mutual normalizability, make some preliminary observations on when it holds, and prove a divergence formula, analogous to a theorem of \cite{Lutz11}, for computing the values of the Billingsley mutual dimensions in many cases.
\section{Mutual Dimension in Cantor Spaces}
In \cite{CasLut15} the authors defined and investigated the mutual dimension between points in Euclidean space. The purpose of this section is to develop a similar framework for the mutual dimension between sequences.

Let $\Sigma = \{0,1,\ldots k-1\}$ be our alphabet and $\Sigma^{\infty}$ denote the set of all $k$-ary sequences over $\Sigma$. For $S,T \in \Sigma^{\infty}$, the notation $(S,T)$ represents the sequence in $(\Sigma \times \Sigma)^{\infty}$ obtained after pairing each symbol in $S$ with the symbol in $T$ located at the same position. For $S \in \Sigma^{\infty}$, let
\begin{equation}\label{real rep}
\alpha_S = \displaystyle\sum\limits_{i=0}^{\infty}S[i]k^{-(i+1)} \in [0,1].
\end{equation}
Informally, we say that $\alpha_S$ is the \emph{real representation} of $S$. Note that, in this section, we often use the notation $S \upharpoonright r$ to mean the first $r \in \mathbb{N}$ symbols of a sequence $S$.

We begin by reviewing some definitions and theorems of algorithmic information theory. All Turing machines are assumed to be self-delimiting.

\begin{definition}
The \emph{conditional Kolmogorov complexity} of $u \in \Sigma^*$ given $w \in \Sigma^*$ with respect to a Turing machine $M$ is
\[
K_M(u\,|\,w) = \min \{|\pi|\, \big|\,\pi \in \{0,1\}^* \text{ and } M(\pi,w)=u\}.
\]
\end{definition}
We define the \emph{Kolmogorov complexity} of $u \in \Sigma^*$ with respect to a Turing machine $M$ by $K_M(u) = K_M(u\,|\,\lambda)$, where $\lambda$ is the \emph{empty string}.

\begin{definition}
A Turing machine $M'$ is \emph{optimal} if, for all Turing machines $M$, there exists a constant $c \in \{0,1\}^*$ such that
\[
K_{M'}(u) \leq K_M(u) + c,
\]
for all $u \in \{0,1\}^*$.
\end{definition}

The following theorem is an important observation in algorithmic information theory.

\begin{theorem}[Optimality]
Every universal Turing machine is optimal.
\end{theorem}

For the duration of this paper, we let $U$ be some fixed universal Turing machine.

\begin{definition}
The \emph{conditional Kolmogorov complexity} of $u \in \Sigma^*$ given $w \in \Sigma^*$ is
\[
K(u\,|\,w) = K_U(u\,|\,w).
\]
\end{definition}

The \emph{Kolmogorov complexity} of a string $u \in \Sigma^*$ is $K(u) = K(u\,|\,\lambda)$. For a detailed overview of Kolmogorov complexity and its properties, see \cite{bLiVit08}.

The following definition of the Kolmogorov complexity of sets of strings is also useful.

\begin{definition}[Definition](Shen and Vereshchagin \cite{jSheVer02}).
The \emph{Kolmogorov complexity} of a set $S \subseteq \Sigma^*$ is
\[
K(S) = \min \{K(u)\,|\,u \in S\}.
\]
\end{definition}

\begin{definition}
The \emph{lower} and \emph{upper dimensions} of $S \in \Sigma^{\infty}$ are
\begin{align*}
dim(S) = \displaystyle\liminf\limits_{u \rightarrow S}\frac{K(u)}{|u| \log |\Sigma|}
\end{align*}
and
\begin{align*}
Dim(S) = \displaystyle\limsup\limits_{u \rightarrow S}\frac{K(u)}{|u| \log |\Sigma|},
\end{align*}
respectively.
\end{definition}

We now proceed to prove several lemmas which describe how the dimensions of sequences and the dimensions of points in Euclidean space correspond to one another.

Keeping in mind that tuples of rationals in $\mathbb{Q}^n$ can be easily encoded as a string in $\Sigma^*$, we use the following definition of the Kolmogorov complexity of points in Euclidean space.

\begin{definition}
The \emph{Kolmogorov complexity} of $x \in \mathbb{R}^n$ at \emph{precision} $r \in \mathbb{N}$ is
\[
K_r(x) = K(B_{2^{-r}}(x) \cap \mathbb{Q}^n).
\]
\end{definition}

We recall a useful corollary from \cite{CasLut15} that is used in the proof of Lemma \ref{joint}.

\begin{corollary}\label{rs to r}
For all $x \in \mathbb{R}^n$ and $r,s \in \mathbb{N}$,
\[
K_{r+s}(x) \leq K_r(x) + o(r).
\]
\end{corollary}

\begin{lemma}\label{joint}
There is a constant $c \in \mathbb{N}$ such that, for all $S,T \in \Sigma^{\infty}$ and $r \in \mathbb{N}$,
\begin{align*}
K((S, T)\upharpoonright r) = K_r(\alpha_S, \alpha_T) + o(r).
\end{align*}
\end{lemma}

\begin{proof}
First we show that $K_r(\alpha_S,\alpha_T) \leq K((S,T) \upharpoonright r) + o(r)$.

Observe that
\begin{align*}
&\,\,\,\,\,\,\, \big| (\alpha_S, \alpha_T) - (\alpha_{S\upharpoonright r}, \alpha_{T\upharpoonright r}) \big|\\
&= \bigg | \bigg (\displaystyle\sum\limits_{i=0}^{\infty}S[i]k^{-(i+1)}, \displaystyle\sum\limits_{i=0}^{\infty}T[i]k^{-(i+1)} \bigg ) - \bigg (\displaystyle\sum\limits_{i=0}^{r-1}S[i]k^{-(i+1)}, \displaystyle\sum\limits_{i=0}^{r-1}S[i]k^{-(i+1)} \bigg ) \bigg|\\
&= \bigg | \bigg (\displaystyle\sum\limits_{i=r}^{\infty}S[i]k^{-(i+1)}, \displaystyle\sum\limits_{i=r}^{\infty}T[i]k^{-(i+1)} \bigg ) \bigg |\\
&\leq \bigg | \bigg (\displaystyle\sum\limits_{i=r}^{\infty}S[i]2^{-(i+1)}, \displaystyle\sum\limits_{i=r}^{\infty}T[i]2^{-(i+1)} \bigg ) \bigg |\\
&=|(2^{-r}, 2^{-r})|\\
&\leq 2^{1-r},
\end{align*}
which implies the inequality
\begin{align}\label{minusone}
K_{r-1}(\alpha_S,\alpha_T) \leq K(\alpha_{S\upharpoonright r}, \alpha_{T\upharpoonright r}).
\end{align}

Let $M$ be a Turing machine such that, if $U(\pi) = (u_0,w_0)(u_1,w_1)\cdots (u_{n-1},w_{n-1}) \in (\Sigma \times \Sigma)^*$,
\begin{align}\label{monpi}
M(\pi) = \bigg( \displaystyle\sum\limits^{n-1}_{i=0}u_i \cdot k^{-(i+1)}, \displaystyle\sum\limits^{n-1}_{i=0}w_i \cdot k^{-(i+1)} \bigg).
\end{align}
Let $c_M$ be an optimality constant for $M$ and $\pi \in \{0,1\}^*$ be a minimum-length program for $(S,T) \upharpoonright r$. By optimality and (\ref{monpi}),
\begin{align}
K(\alpha_{S\upharpoonright r},\alpha_{T\upharpoonright r}) &\leq K_M(\alpha_{S\upharpoonright r},\alpha_{T\upharpoonright r})\notag\\
																													 &\leq |\pi| + c_M \label{alphatoseq}\\
																													 &= K((S,T)\upharpoonright r) + c_M \notag.
\end{align}
Therefore, by Corollary \ref{rs to r}, (\ref{minusone}), and (\ref{alphatoseq}),
\begin{align*}
K_r(\alpha_S,\alpha_T) &\leq K_{r-1}(\alpha_S,\alpha_T) + o(r)\\
											 &\leq K(\alpha_{S\upharpoonright r}, \alpha_{T \upharpoonright r}) + o(r)\\
											 &\leq K((S,T)\upharpoonright r) + o(r).
\end{align*}

Next we prove that $K((S,T)\upharpoonright r) \leq K_r(\alpha_S, \alpha_T) + O(1)$. We consider the case where $S = x(k-1)^{\infty}$, $T \neq y(k-1)^{\infty}$, and $x \in \{0,1\}^*$ and $y \in \{0,1\}^*$ are either empty or end with a symbol other than $(k-1)$, i.e., $S$ has a tail that is an infinite sequence of the largest symbol in $\Sigma$ and $T$ does not. Let $M'$ be a Turing machine such that, if $U(\pi) = \langle q,p \rangle$ for any two rationals $q,p \in [0,1]$,
\begin{align}\label{m'def}
M'(\pi) = (u_0,w_0)(u_1,w_1)\cdots (u_{r-1},w_{r-1}) \in (\Sigma \times \Sigma)^*,
\end{align}
where $M'$ operates by running $\pi$ on $U$ to obtain $(q,p)$ and searching for strings $u = u_0u_1\cdots u_{r-1}$ and $w = w_0w_1\cdots w_{r-1}$ such that
\begin{equation}\label{first}
q = \displaystyle\sum\limits^{|x|-1}_{i=0}u_ik^{-(i+1)} + (k-1)k^{-(|x|+1)}, u_{|x|-1} < (k-1), \text{ and } u_i = (k-1) \text{ for } i \geq |x|,
\end{equation}
and
\begin{equation}\label{second}
w_i \cdot k^{-(i+1)} \leq p - (w_0\cdot k^{-1} + w_1\cdot k^{-2} + \cdots + w_{i-1}\cdot k^{-i}) < (w_i + 1) \cdot k^{-(i+1)}
\end{equation}
for $0 \leq i < r$.

Let $c_{M'}$ be an optimality constant for $M'$ and $m, t \in \mathbb{N}$ such that $m, t \leq k^r - 1$ and
\begin{equation}\label{alpha}
(\alpha_S, \alpha_T) \in [m \cdot k^{-r}, (m+1) \cdot k^{-r}) \times [t \cdot k^{-r}, (t+1) \cdot k^{-r}).
\end{equation}
Let
\begin{equation}\label{rats}
(q,p) \in B_{k^{-r}}(\alpha_S, \alpha_T) \cap [m \cdot k^{-r}, (m+1) \cdot k^{-r}) \times [t \cdot k^{-r}, (t+1) \cdot k^{-r}) \cap \mathbb{Q}^2,
\end{equation}
and let $\pi$ be a minimum-length program for $(q,p)$.
First we show that $u_i = S[i]$ for all $0 \leq i < r$. We do not need to consider the case where $i \geq |x|$ because (\ref{first}) assures us that $u_i = S[i]$. Thus we will always assume that $i < |x|$. If $u_0 \neq S[0]$, then, by (\ref{first}), 
\[
q \notin [S[0]\cdot k^{-1}, (S[0]+1)\cdot k^{-1}).
\]
By (\ref{alpha}), this implies that
\begin{align*}
q \notin [m\cdot k^{-r}, (m+1)\cdot k^{-r}),
\end{align*}
which contradicts (\ref{rats}). Now assume that $u_n = S[n]$ for all $n \leq i < r - 1$. If $u_{i+1} \neq S[i+1]$, then, by (\ref{first}),
\[
q \notin \bigg[\displaystyle\sum\limits_{n=0}^{i}S[n]\cdot k^{-(i+1)} + S[i+1]\cdot k^{-(i+2)}, \displaystyle\sum\limits_{n=0}^{i}S[n]\cdot k^{-(i+1)} + (S[i+1]+1)\cdot k^{-(i+2)} \bigg).
\]
By (\ref{alpha}), this implies that
\begin{align*}
q \notin [m\cdot k^{-r}, (m+1)\cdot k^{-r}),
\end{align*}
which contradicts (\ref{rats}). Therefore, $u_i = S[i]$ for all $0 \leq i < r$. A similar argument shows that $w_i = T[i]$, so we conclude that $M'(q,p) = (S,T) \upharpoonright r$.

By optimality, (\ref{m'def}), and (\ref{rats}),
\begin{align*}
K((S,T)\upharpoonright r) &\leq K_{M'}((S,T)\upharpoonright r) + c_{M'}\\
											&\leq |\pi| + c_{M'}\\
											&= K(q,p) + c_{M'}\\
											&= K(B_{2^{-r}}(\alpha_S, \alpha_T) \cap [0,1]^2) + c_{M'}\\
											&\leq K_r(\alpha_S, \alpha_T) + O(1),
\end{align*}
where the last inequality holds simply because we can design a Turing machine to transform any point from outside the unit square to its edge. All other cases for $S$ and $T$ can be proved in a similar manner. \qedhere
\end{proof}

\begin{lemma}\label{KolEquiv}
There is a constant $c \in \mathbb{N}$ such that, for all $S \in \Sigma^{\infty}$ and $r \in \mathbb{N}$,
\begin{align*}
K(S\upharpoonright r) = K_r(\alpha_S) + c.
\end{align*}
\end{lemma}

\begin{proof}
Let $0^{\infty}$ represent the sequence containing all 0's. It is clear that there exist constants $c_1,c_2 \in \mathbb{N}$ such that
\[
K(S\upharpoonright r) = K((S,0^{\infty})\upharpoonright r) + c_1
\]
and
\[
K_r(\alpha_S, 0) = K_r(\alpha_S) + c_2.
\]
Therefore, by the above inequalities and Lemma \ref{joint},
\begin{align*}
K(S\upharpoonright r) &= K((S,0^{\infty})\upharpoonright r) + c_1\\
											&= K(\alpha_S, 0) + o(r) + c_1\\
											&= K_r(\alpha_S) + o(r) + c_1 + c_2\\
											&= K_r(\alpha_S) + o(r). \qedhere
\end{align*}
\end{proof}

\begin{definition}
For any point $x \in \mathbb{R}^n$, the \emph{lower} and \emph{upper dimensions} of $x$ are
\begin{align*}
dim(x) = \displaystyle\liminf\limits_{r \rightarrow \infty}\frac{K_r(x)}{r}
\end{align*}
and
\begin{align*}
Dim(x) = \displaystyle\limsup\limits_{r \rightarrow \infty}\frac{K_r(x)}{r},
\end{align*}
respectively.
\end{definition}

The next two corollaries describe principles that relate the dimensions of sequences to the dimensions of the sequences' real representations. The first follows from Lemma \ref{joint} and the second follows from Lemma \ref{KolEquiv}.

\begin{corollary}\label{jointDimEquiv}
For all $S, T \in \Sigma^{\infty}$,
\begin{align*}
dim(S,T) = dim(\alpha_S, \alpha_T)\,\,\, \text{and} \,\,\, Dim(S,T) = Dim(\alpha_S, \alpha_T).
\end{align*}
\end{corollary}

\begin{corollary}\label{dimEquiv}
For any sequence $S \in \Sigma^{\infty}$,
\begin{align*}
dim(S) = dim(\alpha_S).
\end{align*}
\end{corollary}

\begin{lemma}\label{yxlemma}
There is a constant $c \in \mathbb{N}$ such that, for all $x,y \in \{0,1\}^*$,
\begin{align*}
K(y\,|\,x) \leq K(y\,|\,\langle x,K(x) \rangle) + K(K(x)) + c.
\end{align*}
\end{lemma}

\begin{proof}
Let $M$ be a Turing machine such that, if $U(\pi_1) = K(x)$ and $U(\pi_2,\langle x,K(x) \rangle) = y$,
\begin{align*}
M(\pi_1\pi_2,x) = y.
\end{align*}
Let $c_M \in \mathbb{N}$ be an optimality constant of $M$. Assume the hypothesis, and let $\pi_1$ be a minimum-length program for $K(x)$ and $\pi_2$ be a minimum-length program for $y$ given $x$ and $K(x)$. By optimality,
\begin{align*}
K(y\,|\,x) &\leq K_M(y\,|\,x) + c_M\\
			 &\leq |\pi_1\pi_2| + c_M\\
			 &= K(y\,|\,\langle x,K(x) \rangle) + K(K(x)) + c,
\end{align*}
where $c = c_M$. \qedhere
\end{proof}

\begin{lemma}\label{olemma}
For all $x \in \{0,1\}^*$, $K(K(x)) = o(|x|)$ as $|x| \rightarrow \infty$.
\end{lemma}

\begin{proof}
There exist constants $c_1,c_2 \in \mathbb{N}$ such that
\begin{align*}
K(K(x)) &\leq \log{K(x)} + c_1\\
				&\leq \log{(|x| + c_2)} + c_1\\
				&= o(|x|).
\end{align*}
as $|x| \rightarrow \infty$.
\end{proof}

The following lemma is well-known and can be found in \cite{bLiVit08}. \qedhere

\begin{lemma}\label{fact1}
There is a constant $c \in \mathbb{N}$ such that, for all $x,y \in \{0,1\}^*$,
\begin{align*}
K(x,y) = K(x) + K(y\,|\, x,K(x)) + c.
\end{align*}
\end{lemma}

The following is a corollary of Lemma \ref{fact1}.

\begin{corollary}\label{fact2}
There is a constant $c \in \mathbb{N}$ such that, for all $x,y \in \{0,1\}^*$,
\begin{align*}
K(x,y) \leq K(x) + K(y\,|\,x) + c.
\end{align*}
\end{corollary}

\begin{lemma}\label{helping lemma}
For all $x,y \in \{0,1\}^*$,
\begin{align*}
K(y\,|\,x) + K(x) \leq K(x,y) + o(|x|) \text{ as } |x| \rightarrow \infty.
\end{align*}
\end{lemma}

\begin{proof}
By Lemma \ref{yxlemma}, there is a constant $c_1 \in \mathbb{N}$ such that
\[
K(y\,|\,x) \leq K(y\,|\,\langle x,K(x) \rangle) + K(K(x)) + c_1.
\]
This implies that
\[
K(y\,|\,x) + K(x) \leq K(y\,|\,\langle x,K(x) \rangle) + K(K(x)) + K(x) + c_1.
\]
By Lemma \ref{fact1}, there is a constant $c_2 \in \mathbb{N}$ such that
\[
K(y\,|\,x) + K(x) \leq K(x,y) + K(K(x)) + c_1 + c_2.
\]
Therefore, by Lemma \ref{olemma},
\[
K(y\,|\,x) + K(x) \leq K(x,y) + o(|x|).
\]
as $|x| \rightarrow \infty$. \qedhere
\end{proof}

The rest of this section is about mutual information and mutual dimension. We now provide the definitions of the mutual information between strings as defined in \cite{bLiVit08} and the mutual dimension between sequences.

\begin{definition}
The (\emph{algorithmic}) \emph{mutual information} between $u \in \Sigma^*$ and $w \in \Sigma^*$ is
\[
I(u:w) = K(w) - K(w\,|\,u).
\]
\end{definition}

\begin{definition}
The \emph{lower} and \emph{upper mutual dimensions} between $S \in \Sigma^{\infty}$ and $T \in \Sigma^{\infty}$ are
\begin{align*}
mdim(S:T) = \displaystyle\liminf\limits_{(u,w) \rightarrow (S, T)}\frac{I(u:w)}{|u| \log |\Sigma|}
\end{align*}
and
\begin{align*}
Mdim(S:T) = \displaystyle\limsup\limits_{(u,w) \rightarrow (S, T)}\frac{I(u:w)}{|u| \log |\Sigma|},
\end{align*}
respectively.
\end{definition}
(We insist that $|u| = |w|$ in the above limits.) The mutual dimension between two sequences is regarded as the density of algorithmic mutual information between them.

\begin{lemma}\label{x:y}
For all strings $x,y \in \{0,1\}^*$,
\begin{align*}
I(x:y) = K(x) + K(y) - K(x,y) + o(|x|).
\end{align*}
\end{lemma}

\begin{proof}
By definition of mutual information and Lemma \ref{helping lemma},
\begin{align*}
I(x:y) &= K(y) - K(y\,|\,x)\\
			 &\geq K(x) + K(y) - K(x,y) + o(|x|).
\end{align*}
as $|x| \rightarrow \infty$. Also, by Corollary \ref{fact2}, there is a constant $c \in \mathbb{N}$ such that
\begin{align*}
I(x:y) &= K(y) - K(y\,|\,x)\\
			 &\leq K(x) + K(y) - K(x,y) + c\\
			 &= K(x) + K(y) - K(x,y) + o(|x|).
\end{align*}
as $|x| \rightarrow \infty$. \qedhere
\end{proof}

The next two definitions were proposed and thoroughly investigated in \cite{CasLut15}.

\begin{definition}
The \emph{mutual information} between $x \in \mathbb{R}^n$ and $y \in \mathbb{R}^m$ at $\emph{precision}$ $r \in \mathbb{N}$ is
\[
I_r(x:y) = \min \{I(q:p)\,|\, q \in B_{2^{-r}}(x) \cap \mathbb{Q}^n \text{ and } p \in B_{2^{-r}}(y) \cap \mathbb{Q}^m\}. 
\]
\end{definition}

\begin{definition}
The \emph{lower} and \emph{upper mutual dimensions} between $x \in \mathbb{R}^n$ and $y \in \mathbb{R}^m$ are
\begin{align*}
mdim(x:y) = \displaystyle\liminf\limits_{r \rightarrow \infty}\frac{I_r(x:y)}{r}
\end{align*}
and
\begin{align*}
Mdim(x:y) = \displaystyle\limsup\limits_{r \rightarrow \infty}\frac{I_r(x:y)}{r}
\end{align*}
\end{definition}

\begin{lemma}\label{correspond}
For all $S,T \in \Sigma^{\infty}$ and $r \in \mathbb{N}$,
\begin{align*}
I(S \upharpoonright r:T\upharpoonright r) = I_r(\alpha_S:\alpha_T) + o(r).
\end{align*}
\end{lemma}

\begin{proof}
By Lemmas \ref{KolEquiv}, \ref{joint}, and \ref{x:y},
\begin{align*}
I(S\upharpoonright r:T\upharpoonright r) &= K(S\upharpoonright r) + K(T\upharpoonright r) - K((S, T)\upharpoonright r) + o(r)\\
																				 &= K_r(\alpha_S) + K_r(\alpha_T) - K_r(\alpha_S,\alpha_T) + o(r)\\
																				 &= I_r(\alpha_S:\alpha_T) + o(r).
\end{align*}
as $r \rightarrow \infty$. \qedhere
\end{proof}

The following corollary follows immediately from Lemma \ref{correspond} and relates the mutual dimension between sequences to the mutual dimension between the sequences' real representations.

\begin{corollary}\label{mdimEquiv}

For all $S,T \in \Sigma^{\infty}$,
\begin{align*}
mdim(S:T) = mdim(\alpha_S:\alpha_T)\,\,\, \text{and} \,\,\, MDim(S:T) = Mdim(\alpha_S:\alpha_T).
\end{align*}
\end{corollary}

Our main result for this section shows that mutual dimension between sequences is well behaved.

\begin{theorem}\label{mdim props}
For all $S,T \in \Sigma^{\infty}$,
\begin{enumerate}
\item $dim(S) + dim(T) - Dim(S,T) \leq mdim(S:T) \leq Dim(S) + Dim(T) - Dim(S,T)$.
\item $dim(S) + dim(T) - dim(S,T) \leq Mdim(S:T) \leq Dim(S) + Dim(T) - dim(S,T)$.
\item $mdim(S:T) \leq \min\{dim(S),dim(T)\}$; $Mdim(S:T) \leq \min\{Dim(S),Dim(T)\}$.
\item $0 \leq mdim(S:T) \leq Mdim(S:T) \leq 1$.
\item $mdim(S:T) = mdim(T:S)$; $Mdim(S:T) = Mdim(T:S)$.
\end{enumerate}
\end{theorem}

\begin{proof}
The theorem follows directly from the properties of mutual dimension between points in Euclidean space found in \cite{CasLut15} and the correspondences described in corollaries \ref{jointDimEquiv}, \ref{dimEquiv}, and \ref{mdimEquiv}. \qedhere
\end{proof}
\section{Mutual Dimension and Coupled Randomness}
In this section we investigate the mutual dimensions between coupled random sequences. Because coupled randomness is new to algorithmic information theory, we first review the technical framework for it. Let $\Sigma$ be a finite alphabet. A (\emph{Borel}) \emph{probability measure} on the Cantor space $\Sigma^{\infty}$ of all infinite sequences over $\Sigma$ is (conveniently represented by) a function $\nu: \Sigma^* \rightarrow [0,1]$ with the following two properties.
\begin{enumerate}
\item $\nu(\lambda) = 1$, where $\lambda$ is the empty string.
\vspace*{2mm}
\item For every $w \in \Sigma^*$, $\nu(w) = \displaystyle\sum\limits_{a \in \Sigma} \nu(wa)$.
\end{enumerate}
Intuitively, here, $\nu(w)$ is the probability that $w \sqsubseteq S$ ($w$ is a \emph{prefix} of $S$) when $S \in \Sigma^{\infty}$ is ``chosen according to'' the probability measure $\nu$.

Most of this paper concerns a very special class of probability measures on $\Sigma^{\infty}$. For each $n \in \mathbb{N}$, let $\alpha^{(n)}$ be a probability measure on $\Sigma$, i.e., $\alpha^{(n)}: \Sigma \rightarrow [0,1]$, with 
\[
\displaystyle\sum\limits_{a \in \Sigma}\alpha^{(n)}(a) = 1,
\]
and let $\vec{\alpha} = (\alpha^{(0)}, \alpha^{(1)}, \ldots)$ be the sequence of these probability measures on $\Sigma$. Then the \emph{product} of $\vec{\alpha}$ (or, emphatically distinguishing it from the products $\nu_1 \times \nu_2$ below, the \emph{longitudinal product} of $\vec{\alpha}$) is the probability measure $\mu[\vec{\alpha}]$ on $\Sigma^{\infty}$ defined by
\[
\mu[\vec{\alpha}](w) = \displaystyle\prod\limits_{n=0}^{|w|-1}\alpha^{(n)}(w[n])
\]
for all $w \in \Sigma^*$, where $w[n]$ is the $n^{th}$ symbol in $w$. Intuitively, a sequence $S \in \Sigma^{\infty}$ is ``chosen according to'' $\mu[\vec{\alpha}]$ by performing the successive experiments $\alpha^{(0)}, \alpha^{(1)}, \ldots$ \emph{independently}.

To extend probability to pairs of sequences, we regard $\Sigma \times \Sigma$ as an alphabet and rely on the natural identification between $\Sigma^{\infty} \times \Sigma^{\infty}$ and $(\Sigma \times \Sigma)^{\infty}$. A probability measure on $\Sigma^{\infty} \times \Sigma^{\infty}$ is thus a function $\nu: (\Sigma \times \Sigma)^* \rightarrow [0,1]$. It is convenient to write elements of $(\Sigma \times \Sigma)^*$ as ordered pairs $(u,v)$, where $u,v \in \Sigma^*$ \emph{have the same length}. With this notation, condition 2 above says that, for every $(u,v) \in (\Sigma \times \Sigma)^*$,
\[
\nu(u,v) = \displaystyle\sum\limits_{a,b \in \Sigma}\nu(ua,vb).
\]

If $\nu$ is a probability measure on $\Sigma^{\infty} \times \Sigma^{\infty}$, then the first and second \emph{marginal probability measures} of $\nu$ (briefly, the first and second \emph{marginals} of $\nu$) are the functions $\nu_1,\nu_2: \Sigma^* \rightarrow [0,1]$ defined by
\[
\nu_1(u) = \displaystyle\sum\limits_{v \in \Sigma^{|u|}}\nu(u,v),\,\,\,\nu_2(v) = \displaystyle\sum\limits_{u \in \Sigma^{|v|}}\nu(u,v).
\]
It is easy to verify that $\nu_1$ and $\nu_2$ are probability measures on $\Sigma^*$. The probability measure $\nu$ here is often called a \emph{joint probability measure} on $\Sigma^{\infty} \times \Sigma^{\infty}$, or a \emph{coupling} of the probability measures $\nu_1$ and $\nu_2$.

If $\nu_1$ and $\nu_2$ are probability measures on $\Sigma^{\infty}$, then the \emph{product probability measure} $\nu_1 \times \nu_2$ on $\Sigma^{\infty} \times \Sigma^{\infty}$ is defined by
\[
(\nu_1 \times \nu_2)(u,v) = \nu_1(u)\nu_2(v)
\]
for all $u,v \in \Sigma^*$ with $|u| = |v|$. It is well known and easy to see that $\nu_1 \times \nu_2$ is, indeed a probability measure on $\Sigma^{\infty} \times \Sigma^{\infty}$ and that the marginals of $\nu_1 \times \nu_2$ are $\nu_1$ and $\nu_2$. Intuitively, $\nu_1 \times \nu_2$ is the coupling of $\nu_1$ and $\nu_2$ in which $\nu_1$ and $\nu_2$ are \emph{independent}, or \emph{uncoupled}.

We are most concerned here with coupled longitudinal product probability measures on $\Sigma^{\infty} \times \Sigma^{\infty}$. For each $n \in \mathbb{N}$, let $\alpha^{(n)}$ be a probability measure on $\Sigma \times \Sigma$, i.e., $\alpha^{(n)}: \Sigma \times \Sigma \rightarrow [0,1]$, with
\[
\displaystyle\sum\limits_{a,b \in \Sigma}\alpha^{(n)}(a,b) = 1,
\]
and let $\vec{\alpha} = (\alpha^{(0)}, \alpha^{(1)}, \ldots)$ be the sequence of these probability measures. Then the longitudinal product $\mu[\vec{\alpha}]$ is defined as above, but now treating $\Sigma \times \Sigma$ as the alphabet. It is easy to see that the marginals of $\mu[\vec{\alpha}]$ are $\mu[\vec{\alpha}]_1 = \mu[\vec{\alpha_1}]$ and $\mu[\vec{\alpha}]_2 = \mu[\vec{\alpha_2}]$, where each $\alpha_i^{(n)}$ is the marginal on $\Sigma$ given by
\[
\alpha_1^{(n)}(a) = \displaystyle\sum\limits_{b \in \Sigma}\alpha^{(n)}(a,b),\,\,\,\alpha_2^{(n)}(b) = \displaystyle\sum\limits_{a \in \Sigma}\alpha^{(n)}(a,b).
\]

The following class of examples is useful \cite{ODon14} and instructive.

\begin{expl}\label{murho} \normalfont{Let $\Sigma = \{0,1\}$. For each $n \in \mathbb{N}$, fix a real number $\rho_n \in [-1,1]$, and define the probability measure $\alpha^{(n)}$ on $\Sigma \times \Sigma$ by $\alpha^{(n)}(0,0) = \alpha^{(n)}(1,1) = \frac{1+\rho_n}{4}$ and $\alpha^{(n)}(0,1) = \alpha^{(n)}(1,0) = \frac{1-\rho_n}{4}$. Then, writing $\alpha^{\vec{\rho}}$ for $\vec{\alpha}$, the longitudinal product $\mu[\alpha^{\vec{\rho}}]$ is a probability measure on $\bf{C} \times \bf{C}$. It is routine to check that the marginals of $\mu[\alpha^{\vec{\rho}}]$ are
\[
\mu[\alpha^{\vec{\rho}}]_1 = \mu[\alpha^{\vec{\rho}}]_2 = \mu,
\]
where $\mu(w) = 2^{-|w|}$ is the uniform probability measure on \bf{C}.}
\end{expl}

It is convenient here to use Schnorr's martingale characterization \cite{Schn71a,Schn71b,Schn77,bLiVit08,Nies09,DowHir10} of the algorithmic randomness notion introduced by Martin-L{\"o}f \cite{Mart66}. If $\nu$ is a probability measure on $\Sigma^{\infty}$, then a $\nu$--\emph{martingale} is a function $d: \Sigma^* \rightarrow [0, \infty)$ satisfying $d(w)\nu(w) = \sum_{a \in \Sigma}d(wa)\nu(wa)$ for all $w \in \Sigma^*$. A $\nu$--martingale $d$ \emph{succeeds} on a sequence $S \in \Sigma^{\infty}$ if $\limsup_{w \rightarrow S}d(w) = \infty$. A $\nu$--martingale $d$ is \emph{constructive}, or \emph{lower semicomputable}, if there is a computable function $\hat{d}: \Sigma^* \times \mathbb{N} \rightarrow \mathbb{Q} \cap [0, \infty]$ such that $\hat{d}(w,t) \leq \hat{d}(w,t+1)$ holds for all $w \in \Sigma^*$ and $t \in \mathbb{N}$, and $\lim_{t \rightarrow \infty}\hat{d}(w,t) = d(w)$ holds for all $w \in \Sigma^*$. A sequence $R \in \Sigma^{\infty}$ is \emph{random} with respect to a probability measure $\nu$ on $\Sigma^*$ if no lower semicomputable $\nu$--martingale succeeds on $R$.

If we once again treat $\Sigma \times \Sigma$ as an alphabet, then the above notions all extend naturally to $\Sigma^{\infty} \times \Sigma^{\infty}$. Hence, when we speak of a \emph{coupled pair} $(R_1,R_2)$ \emph{of random sequences}, we are referring to a pair $(R_1,R_2) \in \Sigma^{\infty} \times \Sigma^{\infty}$ that is random with respect to some probability measure $\nu$ on $\Sigma^{\infty} \times \Sigma^{\infty}$ that is explicit or implicit in the discussion. An extensively studied special case here is that $R_1, R_2 \in \Sigma^{\infty}$ are defined to be \emph{independently random} with respect to probability measures $\nu_1, \nu_2$, respectively, on $\Sigma^{\infty}$ if $(R_1,R_2)$ is random with respect to the product probability measure $\nu_1 \times \nu_2$ on $\Sigma^{\infty} \times \Sigma^{\infty}$.

When there is no possibility of confusion, we use such convenient abbreviations as ``random with respect to $\vec{\alpha}$'' for ``random with respect to $\mu[\vec{\alpha}]$.''

A trivial transformation of Martin-L{\"o}f tests establishes the following well known fact.

\begin{observation}\label{rand marg}
If $\nu$ is a computable probability measure on $\Sigma^{\infty} \times \Sigma^{\infty}$ and $(R_1, R_2) \in \Sigma^{\infty} \times \Sigma^{\infty}$ is random with respect to $\nu$, then $R_1$ and $R_2$ are random with respect to the  marginals $\nu_1$ and $\nu_2$.
\end{observation}

\begin{expl}
\normalfont{If $\vec{\rho}$ is a computable sequence of reals $\rho_n \in [-1,1]$, $\alpha^{\vec{\rho}}$ is as in Example \ref{murho}, and $(R_1,R_2) \in \bf{C} \times \bf{C}$ is random with respect to $\alpha^{\vec{\rho}}$, then Observation \ref{rand marg} tells us that $R_1$ and $R_2$ are random with respect to the uniform probability measure on $\bf{C}$.}
\end{expl}

We recall basic definitions from Shannon information theory.

\begin{definition}
Let $\alpha$ be a probability measure on $\Sigma$. The \emph{Shannon entropy} of $\alpha$ is
\[
\mathcal{H}(\alpha) = \displaystyle\sum\limits_{a \in \Sigma} \alpha(a)\log \frac{1}{\alpha(a)}.
\]
\end{definition}

\begin{definition}
Let $\alpha$ be probability measures on $\Sigma \times \Sigma$. The \emph{Shannon mutual information} between $\alpha_1$ and $\alpha_2$ is
\[
I(\alpha_1:\alpha_2) = \displaystyle\sum\limits_{(a,b) \in \Sigma \times \Sigma}\alpha(a,b)\log \frac{\alpha(a,b)}{\alpha_1(a)\alpha_2(b)}.
\]
\end{definition}

\begin{theorem}[\cite{jLutz03a}]\label{dim to H}
If $\vec{\alpha}$ is a computable sequence of probability measures $\alpha^{(n)}$ on $\Sigma$ that converge to a probability measure $\alpha$ on $\Sigma$, then for every $R \in \Sigma^{\infty}$ that is random with respect to $\vec{\alpha}$,
\[
dim(R) = \frac{\mathcal{H}(\alpha)}{\log |\Sigma|}
\]
\end{theorem}
The following is a corollary to Theorem \ref{dim to H}.

\begin{corollary}\label{K to H}
If $\vec{\alpha}$ is a computable sequence of probability measures $\alpha^{(n)}$ on $\Sigma$ that converge to a probability measure $\alpha$ on $\Sigma$, then for every $R \in \Sigma^{\infty}$ that is random with respect to $\vec{\alpha}$ and every $w \sqsubseteq R$,
\[
K(w) = |w|\mathcal{H}(\alpha) + o(|w|).
\]
\end{corollary}

\begin{lemma}\label{I to I}
If $\vec{\alpha}$ is a computable sequence of probability measures $\alpha^{(n)}$ on $\Sigma \times \Sigma$ that converge to a probability measure $\alpha$ on $\Sigma \times \Sigma$, then for every coupled pair $(R_1,R_2) \in \Sigma^{\infty} \times \Sigma^{\infty}$ that is random with respect to $\vec{\alpha}$ and $(u,w) \sqsubseteq (R_1,R_2)$,
\[
I(u:w) = |u|I(\alpha_1:\alpha_2) + o(|u|).
\]
\end{lemma}

\begin{proof}
By Lemma \ref{x:y},
\[
I(u:w) = K(u) + K(w) - K(u,w) + o(|u|).
\]
We then apply Observation \ref{rand marg} and Corollary \ref{K to H} to obtain
\begin{align*}
I(u:w) &= |u|(\mathcal{H}(\alpha_1) + \mathcal{H}(\alpha_2) - \mathcal{H}(\alpha)) + o(|u|)\\
			 &= |u|I(\alpha_1:\alpha_2) + o(|u|). \qedhere
\end{align*}
\end{proof}
The following is a corollary to Lemma \ref{I to I}.

\begin{corollary}\label{i to i}
If $\alpha$ is a computable, positive probability measure on $\Sigma \times \Sigma$, then, for every sequence $(R_1,R_2) \in \Sigma^{\infty} \times \Sigma^{\infty}$ that is random with respect to $\alpha$ and $(u,w) \sqsubseteq (R_1,R_2)$,
\[
I(u:w) = |u|I(\alpha_1:\alpha_2) + o(|u|).
\]
\end{corollary}

In applications one often encounters longitudinal product measures $\mu[\vec{\alpha}]$ in which the probability measures $\alpha^{(n)}$ are all the same (the i.i.d. case) or else converge to some limiting probability measure. The following theorem says that, in such cases, the mutual dimensions of coupled pairs of random sequences are easy to compute.

\begin{theorem}\label{random resp}
If $\vec{\alpha}$ is a computable sequence of probability measures $\alpha^{(n)}$ on $\Sigma \times \Sigma$ that converge to a probability measure $\alpha$ on $\Sigma \times \Sigma$, then for every coupled pair $(R_1,R_2) \in \Sigma^{\infty} \times \Sigma^{\infty}$ that is random with respect to $\vec{\alpha}$,
\[
mdim(R_1:R_2) = Mdim(R_1:R_2) = \frac{I(\alpha_1:\alpha_2)}{\log |\Sigma|}.
\]
\end{theorem}

\begin{proof}
By Lemma \ref{I to I}, we have
\begin{align*}
mdim(R_1:R_2) &= \displaystyle\liminf\limits_{(u,w) \rightarrow (R_1,R_2)}\frac{I(u:w)}{|u|\log |\Sigma|}\\
							&= \displaystyle\liminf\limits_{(u,w) \rightarrow (R_1,R_2)}\frac{|u|I(\alpha_1:\alpha_2) + o(|u|)}{|u|\log |\Sigma|}\\
							&= \frac{I(\alpha_1:\alpha_2)}{\log |\Sigma|}					
\end{align*}
A similar proof shows that $Mdim(R_1:R_2) = I(\alpha_1:\alpha_2)$. \qedhere
\end{proof}

\begin{expl}\label{rho example}
\normalfont{Let $\Sigma = \{0,1\}$, and let $\vec{\rho}$ be a computable sequence of reals $\rho_n \in [-1,1]$ that converge to a limit $\rho$. Define the probability measure $\alpha$ on $\Sigma \times \Sigma$ by $\alpha(0,0) = \alpha(1,1) = \frac{1+\rho}{4}$ and $\alpha(0,1) = \alpha(1,0) = \frac{1-\rho}{4}$, and let $\alpha_1$ and $\alpha_2$ be the marginals of $\alpha$. If $\alpha^{\vec{\rho}}$ is as in Example \ref{murho}, then for every pair $(R_1,R_2) \in \Sigma^{\infty} \times \Sigma^{\infty}$ that is random with respect to $\alpha^{\vec{\rho}}$, Theorem \ref{random resp} tells us that}
\begin{align*}
mdim(R_1:R_2) &= Mdim(R_1:R_2)\\
							&= I(\alpha_1:\alpha_2)\\
							&= 1 - \mathcal{H}(\frac{1+\rho}{2}).
\end{align*}
In particular, if the limit $\rho$ is 0, then
\[
mdim(R_1:R_2) = Mdim(R_1:R_2) = 0.
\]
\end{expl}

Theorem \ref{random resp} has the following easy consequence, which generalizes the last sentence of Example \ref{rho example}.

\begin{corollary}\label{mdim seq zero}
If $\vec{\alpha}$ is a computable sequence of probability measures $\alpha^{(n)}$ on $\Sigma \times \Sigma$ that converge to a product probability measure $\alpha_1 \times \alpha_2$ on $\Sigma \times \Sigma$, then for every coupled pair $(R_1,R_2) \in \Sigma^{\infty} \times \Sigma^{\infty}$ that is random with respect to $\vec{\alpha}$,
\[
mdim(R_1:R_2) = Mdim(R_1:R_2) = 0.
\]
\end{corollary}

Applying Corollary \ref{mdim seq zero} to a constant sequence $\vec{\alpha}$ in which each $\alpha^{(n)}$ is a product probability measure $\alpha_1 \times \alpha_2$ on $\Sigma \times \Sigma$ gives the following.

\begin{corollary}\label{mdim zero}
If $\alpha_1$ and $\alpha_2$ are computable probability measures on $\Sigma$, and if $R_1,R_2 \in \Sigma^{\infty}$ are independently random with respect to $\alpha_1,\alpha_2$, respectively, then
\[
mdim(R_1:R_2) = Mdim(R_1:R_2) = 0.
\]
\end{corollary}

We conclude this section by showing that the converse of Corollary \ref{mdim zero} does \emph{not} hold. This can be done via a direct construction, but it is more instructive to use a beautiful theorem of Kakutani, van Lambalgen, and Vovk. The \emph{Hellinger distance} between two probability measures $\alpha_1$ and $\alpha_2$ on $\Sigma$ is
\[
H(\alpha_1,\alpha_2) = \sqrt{\displaystyle\sum\limits_{a \in \Sigma}(\sqrt{\alpha_1(a)} - \sqrt{\alpha_2(a)})^2}.
\]
(See \cite{LePeWi09}, for example.) A sequence $\alpha = (\alpha^{(0)}, \alpha^{(1)}, \ldots)$ of probability measures on $\Sigma$ is \emph{strongly positive} if there is a real number $\delta > 0$ such that, for all $n \in \mathbb{N}$ and $a \in \Sigma$, $\alpha^{(n)}(a) \geq \delta$. Kakutani \cite{Kaku48} proved the classical, measure-theoretic version of the following theorem, and van Lambalgen \cite{tLamb87,Lamb87} and Vovk \cite{Vovk87} extended it to algorithmic randomness.

\begin{theorem}\label{helling}
Let $\vec{\alpha}$ and $\vec{\beta}$ be computable, strongly positive sequences of probability measures on $\Sigma$.
\begin{enumerate}
\item If
\[
\displaystyle\sum\limits_{n=0}^{\infty}H(\alpha^{(n)},\beta^{(n)})^2 < \infty,
\]
then a sequence $R \in \Sigma^{\infty}$ is random with respect to $\vec{\alpha}$ if and only if it is random with respect to $\vec{\beta}$.
\item If\
\[
\displaystyle\sum\limits_{n=0}^{\infty}H(\alpha^{(n)},\beta^{(n)})^2 = \infty,
\]
then no sequence is random with respect to both $\vec{\alpha}$ and $\vec{\beta}$.
\end{enumerate}
\end{theorem}

\begin{observation}\label{coup helling}
Let $\Sigma = \{0,1\}$. If $\rho = [-1,1]$ and probability measure $\alpha$ on $\Sigma \times \Sigma$ is defined from $\rho$ as in Example \ref{rho example}, then
\[
H(\alpha_1 \times \alpha_2, \alpha)^2 = 2 - \sqrt{1+\rho} - \sqrt{1-\rho}.
\]
\end{observation}

\begin{proof}
Assume the hypothesis. Then
\begin{align*}
H(\alpha_1 \times \alpha_2, \alpha)^2 &= \displaystyle\sum\limits_{a,b \in \{0,1\}}(\sqrt{\alpha_1(a)\alpha_2(b)} - \sqrt{\alpha(a,b)})^2\\
																			&= \displaystyle\sum\limits_{a,b \in \{0,1\}}\bigg( \frac{1}{2} - \sqrt{\alpha(a,b)} \bigg)^2\\
																			&= 2\bigg( \frac{1}{2} - \sqrt{\frac{1+\rho}{4}} \bigg)^2 + 2 \bigg(\frac{1}{2} - \sqrt{\frac{1-\rho}{4}} \bigg)^2\\
																			&= 2 - \sqrt{1+\rho} - \sqrt{1-\rho}. \qedhere
\end{align*}
\end{proof}

\begin{corollary}\label{not ind}
Let $\Sigma = \{0,1\}$ and $\delta \in (0,1)$. Let $\vec{\rho}$ be a computable sequence of real numbers $\rho_n \in [\delta-1,1-\delta]$, and let $\alpha^{\vec{\rho}}$ be as in Example \ref{murho}. If
\[
\displaystyle\sum\limits_{n=0}^{\infty}\rho_n^2 = \infty,
\]
and if $(R_1,R_2) \in \Sigma^{\infty} \times \Sigma^{\infty}$ is random with respect to $\alpha^{\vec{\rho}}$, then $R_1$ and $R_2$ are not independently random with respect to the uniform probability measure on $\bf{C}$.
\end{corollary}

\begin{proof}
This follows immediately from Theorem \ref{helling}, Observation \ref{coup helling}, and the fact that
\[
\sqrt{1+x} + \sqrt{1-x} = 2 - \frac{x^2}{2} + o(x^2)
\]
as $x \rightarrow 0$. \qedhere
\end{proof}

\begin{corollary}\label{zero not ind}
There exist sequences $R_1,R_2 \in \bf{C}$ that are random with respect to the uniform probability measure on $\bf{C}$ and satisfy $Mdim(R_1:R_2) = 0$, but are not independently random.
\end{corollary}

\begin{proof}
For each $n \in \mathbb{N}$, let
\[
\rho_n = \frac{1}{\sqrt{n+2}}.
\]
Let $\vec{\rho} = (\rho_0,\rho_1, \ldots)$, let $\alpha^{\vec{\rho}}$ be as in Example \ref{murho}, and let $(R_1,R_2) \in \Sigma^{\infty} \times \Sigma^{\infty}$ be random with respect to $\alpha^{\vec{\rho}}$. Observation \ref{rand marg} tells us that $R_1$ and $R_2$ are random with respect to the marginals of $\alpha^{\vec{\rho}}$, both of which are the uniform probability measure on $\bf{C}$. Since $\rho_n \rightarrow 0$ as $n \rightarrow \infty$, the last sentence in Example \ref{rho example} tells us (via Theorem \ref{random resp}) that $Mdim(R_1:R_2) = 0$. Since
\[
\displaystyle\sum\limits_{n=0}^{\infty}\rho_n^2 = \displaystyle\sum\limits_{n=0}^{\infty}\frac{1}{n+2} = \infty,
\]
Corollary \ref{not ind} tells us that $R_1$ and $R_2$ are not independently random. \qedhere
\end{proof}
\section{Billingsley Mutual Dimensions}
We begin this section by reviewing the Billingsley generalization of constructive dimension, i.e., dimension with respect to strongly positive probability measures. A probability measure $\beta$ on $\Sigma^{\infty}$ is \emph{strongly positive} if there exists $\delta > 0$ such that, for all $w \in \Sigma^*$ and $a \in \Sigma$, $\beta(wa) > \delta\beta(w)$.

\begin{definition}
The \emph{Shannon self-information} of $w \in \Sigma$ is
\[
\ell_{\beta}(w) = \displaystyle\sum\limits_{i=0}^{|w|-1}\log \frac{1}{\beta(w[i])}.
\]
\end{definition}

In \cite{jLutMay08}, Lutz and Mayordomo defined (and usefully applied) constructive Billingsley dimension in terms of gales and proved that it can be characterized using Kolmogorov complexity. Since Kolmogorov complexity is more relevant in this discussion, we treat the following theorem as a definition.

\begin{definition}[Definition](Lutz and Mayordomo \cite{jLutMay08}).
The \emph{dimension} of $S \in \Sigma^{\infty}$ \emph{with respect to} a strongly positive probability measure $\beta$ on $\Sigma^{\infty}$ is
\[
dim^{\beta}(S) = \displaystyle\liminf\limits_{w \rightarrow S}\frac{K(w)}{\ell_{\beta}(w)}.
\]
\end{definition}

In the above definition the denominator $\ell_{\beta}(w)$ normalizes the dimension to be a real number in $[0,1]$. It seems natural to define the Billingsley generalization of mutual dimension in a similar way by normalizing the algorithmic mutual information between $u$ and $w$ by $\log \frac{\beta(u,w)}{\beta_1(u)\beta_2(w)}$ (i.e., the \emph{self-mutual information} or \emph{pointwise mutual information} between $u$ and $w$ \cite{bHanKob07}) as $(u,w) \rightarrow (S,T)$. However, this results in bad behavior. For example, the mutual dimension between any two sequences with respect to the uniform probability measure on $\Sigma \times \Sigma$ is \emph{always} undefined. Other thoughtful modifications to this natural definition results in sequences having negative or infinitely large mutual dimension. The main problem here is that, given a particular probability measure, one can construct certain sequences whose prefixes have extremely large positive or negative self-mutual information. In order to avoid undesirable behavior, we restrict the definition of Billingsley mutual dimension to sequences that are mutually normalizable.

\begin{definition}
Let $\beta$ be a probability measure on $\Sigma^{\infty} \times \Sigma^{\infty}$. Two sequences $S,T \in \Sigma^{\infty}$ are \emph{mutually} $\beta$--\emph{normalizable} (in this order) if
\[
\displaystyle\lim\limits_{(u,w) \rightarrow (S,T)}\frac{\ell_{\beta_1}(u)}{\ell_{\beta_2}(w)} = 1.
\]
\end{definition}

\begin{definition}
Let $S,T \in \Sigma^{\infty}$ be mutually $\beta$--normalizable. The \emph{upper} and \emph{lower mutual dimensions} between $S$ and $T$ with respect to $\beta$ are
\begin{displaymath}
mdim^{\beta}(S:T) = \displaystyle\liminf_{(u,w) \rightarrow (S,T)}\frac{I(u:w)}{\ell_{\beta_1}(u)} = \displaystyle\liminf_{(u,w) \rightarrow (S,T)}\frac{I(u:w)}{\ell_{\beta_2}(w)}
\end{displaymath}
and
\begin{displaymath}
Mdim^{\beta}(S:T) = \displaystyle\limsup_{(u,w) \rightarrow (S,T)}\frac{I(u:w)}{\ell_{\beta_1}(u)} = \displaystyle\limsup_{(u,w) \rightarrow (S,T)}\frac{I(u:w)}{\ell_{\beta_2}(w)},
\end{displaymath}
respectively.
\end{definition}
The above definition has nice properties because $\beta$--normalizable sequences have prefixes with asymptotically equivalent self-information. Given the basic properties of mutual information and Shannon self-information, we can see that
\[
0 \leq mdim^{\beta}(S:T) \leq \min \{dim^{\beta_1}(S), dim^{\beta_2}(T)\} \leq 1.
\]
Clearly, $Mdim^{\beta}$ also has a similar property.

\begin{definition}
Let $\alpha$ and $\beta$ be probability measure on $\Sigma$. The \emph{Kullback-Leibler divergence} between $\alpha$ and $\beta$ is
\[
\mathcal{D}(\alpha||\beta) = \displaystyle\sum\limits_{a \in \Sigma}\alpha(a)\log \frac{\alpha(a)}{\beta(a)}
\]
\end{definition}

The following lemma is useful when proving Lemma \ref{equiv norm} and Theorem \ref{mdf}.

\begin{lemma}[Frequency Divergence Lemma \cite{Lutz11}]
If $\alpha$ and $\beta$ are positive probability measures on $\Sigma$, then, for all $S \in FREQ^{\alpha}$,
\[
\ell_{\beta}(w) = (\mathcal{H}(\alpha) + \mathcal{D}(\alpha || \beta))|w| + o(|w|)
\]
as $w \rightarrow S$.
\end{lemma}

The rest of this paper is primarily concerned with probability measures on alphabets. Our first result of this section is a mutual divergence formula for random, mutually $\beta$--normalizable sequences. This can be thought of as a ``mutual'' version of a divergence formula in \cite{Lutz11}.

\begin{theorem}[Mutual Divergence Formula]\label{mdf}
If $\alpha$ and $\beta$ are computable, positive probability measures on $\Sigma \times \Sigma$, then, for every $(R_1,R_2) \in \Sigma^{\infty} \times \Sigma^{\infty}$ that is random with respect to $\alpha$ such that $R_1$ and $R_2$ are mutually $\beta$--normalizable,
\[
mdim^{\beta}(R_1:R_2) {=} Mdim^{\beta}(R_1:R_2) {=} \frac{I(\alpha_1:\alpha_2)}{\mathcal{H}(\alpha_1) + \mathcal{D}(\alpha_1 || \beta_1)} {=} \frac{I(\alpha_1:\alpha_2)}{\mathcal{H}(\alpha_2) + \mathcal{D}(\alpha_2 || \beta_2)}.
\]
\end{theorem}

\begin{proof}
By Corollary \ref{i to i} and the Frequency Divergence Lemma, we have
\begin{align*}
mdim^{\beta}(R_1:R_2) &= \displaystyle\liminf_{(u,w) \rightarrow (R_1,R_2)}\frac{I(u:w)}{\ell_{\beta_1}(u)}\\
									&= \displaystyle\liminf_{(u,w) \rightarrow (R_1,R_2)}\frac{|u|I(\alpha_1:\alpha_2) + o(|u|\log |\Sigma|)}{(\mathcal{H}(\alpha_1) + \mathcal{D}(\alpha_1 || \beta_1))|u| + o(|u|)}\\
									&=\displaystyle\liminf_{(u,w) \rightarrow (R_1,R_2)}\frac{|u|(I(\alpha_1:\alpha_2) + o(\log |\Sigma|))}{|u|((\mathcal{H}(\alpha_1) + \mathcal{D}(\alpha_1 || \beta_1)) + o(1))}\\
									&= \frac{I(\alpha_1:\alpha_2)}{\mathcal{H}(\alpha_1) + \mathcal{D}(\alpha_1 || \beta_1)}.
\end{align*}
Similar arguments show that
\[
mdim^{\beta}(R_1:R_2) = \frac{I(\alpha_1:\alpha_2)}{\mathcal{H}(\alpha_2) + \mathcal{D}(\alpha_2 || \beta_2)}
\]
and
\[
Mdim^{\beta}(R_1:R_2) = \frac{I(\alpha_1:\alpha_2)}{\mathcal{H}(\alpha_1) + \mathcal{D}(\alpha_1 || \beta_1)} = \frac{I(\alpha_1:\alpha_2)}{\mathcal{H}(\alpha_2) + \mathcal{D}(\alpha_2 || \beta_2)}. \qedhere
\]
\end{proof}

We conclude this section by making some initial observations regarding when mutual normalizability can be achieved.

\begin{definition} Let $\alpha_1$, $\alpha_2$, $\beta_1$, $\beta_2$ be probability measures over $\Sigma$. We say that $\alpha_1$ is $(\beta_1,\beta_2)$--\emph{equivalent} to $\alpha_2$ if
\[
\displaystyle\sum\limits_{a \in \Sigma}\alpha_1(a)\log \frac{1}{\beta_1(a)} = \displaystyle\sum\limits_{a \in \Sigma}\alpha_2(a)\log \frac{1}{\beta_2(a)}.
\]
\end{definition}

For a probability measure $\alpha$ on $\Sigma$, let $FREQ_{\alpha}$ be the set of sequences $S \in \Sigma^{\infty}$ satisfying $\lim_{n \rightarrow \infty}n^{-1}|\{i < n \,\big| \, S[i] = a\}| = \alpha(a)$ for all $a \in \Sigma$.

\begin{lemma}\label{equiv norm}
Let $\alpha_1$, $\alpha_2$, $\beta_1$, $\beta_2$ be probability measures on $\Sigma$. If $\alpha_1$ is $(\beta_1,\beta_2)$--equivalent to $\alpha_2$, then, for all pairs $(S,T) \in FREQ_{\alpha_1} \times FREQ_{\alpha_2}$, $S$ and $T$ are mutually $\beta$--normalizable.
\end{lemma}

\begin{proof}
By the Frequency Divergence Lemma,
\begin{align*}
\displaystyle\lim\limits_{(u,w) \rightarrow (S,T)}\frac{\ell_{\beta_1}(u)}{\ell_{\beta_2}(w)} &= \displaystyle\lim\limits_{n \rightarrow \infty}\frac{(\mathcal{H}(\alpha_1) + \mathcal{D}(\alpha_1 || \beta_1)) \cdot n + o(n)}{(\mathcal{H}(\alpha_2) + \mathcal{D}(\alpha_2 || \beta_2)) \cdot n + o(n)}\\
&= \frac{\mathcal{H}(\alpha_1) + \mathcal{D}(\alpha_1 || \beta_1)}{\mathcal{H}(\alpha_2) + \mathcal{D}(\alpha_2 || \beta_2)}\\
&= \frac{\displaystyle\sum\limits_{a \in \Sigma}\alpha_1(a)\log \frac{1}{\beta_1(a)}}{\displaystyle\sum\limits_{a \in \Sigma}\alpha_2(a)\log \frac{1}{\beta_2(a)}}\\
&= 1,
\end{align*}
where the last equality is due to $\alpha_1$ being $(\beta_1,\beta_2)$--equivalent to $\alpha_2$. \qedhere
\end{proof}

Given probability measures $\beta_1$ and $\beta_2$ on $\Sigma$, we would like to know which sequences are mutually $\beta$--normalizable. The following results help to answer this question for probability measures on and sequences over $\{0,1\}$.

\begin{lemma}\label{prob meas}Let $\beta_1$ and $\beta_2$ be probability measures on $\{0,1\}$ such that exactly one of the following conditions hold.
\begin{enumerate}
\item $0 < \beta_2(0) < \beta_1(1) < \beta_1(0) < \beta_2(1) < 1$
\item $0 < \beta_2(1) < \beta_1(0) < \beta_1(1) < \beta_2(0) < 1$
\item $0 < \beta_2(0) < \beta_1(0) < \beta_1(1) < \beta_2(1) < 1$
\item $0 < \beta_2(1) < \beta_1(1) < \beta_1(0) < \beta_2(0) < 1$
\item $\beta_1 = \mu$ and $\beta_2 \neq \mu$.
\end{enumerate}

If $f$ is defined by
\[
f(x) = \frac{x \cdot \log \frac{\beta_1(1)}{\beta_1(0)} + \log \frac{\beta_2(1)}{\beta_1(1)}}{\log \frac{\beta_2(1)}{\beta_2(0)}},
\]
then
\[
0 < f(x) < 1,
\]
for all $x \in [0,1]$.
\end{lemma}

\begin{proof}
First, observe that $f$ is linear and has a negative slope under conditions 1 and 2, a positive slope under conditions 3 and 4, and zero slope under condition 5. We verify that, for all $x \in [0,1]$, $f(x) \in (0,1)$ under each condition.

Under condition 1, we assume
\[
\beta_2(0) < \beta_1(1) < \beta_2(1),
\]
which implies that
\[
\log \frac{\beta_2(0)}{\beta_2(1)} < \log \frac{\beta_1(1)}{\beta_2(1)} < 0.
\]
From the above inequality, we obtain
\[
0 < \frac{\log \frac{\beta_2(1)}{\beta_1(1)}}{\log \frac{\beta_2(1)}{\beta_2(0)}} < 1.
\]
Therefore, by the definition of $f$,
\begin{align}\label{between}
0 < f(0) < 1.
\end{align}


Under the same condition, we have
\[
\beta_1(0) < \beta_2(1),
\]
which implies that
\[
\log \frac{\beta_1(0)}{\beta_1(1)} < \log \frac{\beta_2(1)}{\beta_1(1)}.
\]
From the above inequality, we obtain
\[
\frac{\log \frac{\beta_1(0)}{\beta_1(1)}}{\log \frac{\beta_2(1)}{\beta_2(0)}} < \frac{\log \frac{\beta_2(1)}{\beta_1(1)}}{\log \frac{\beta_2(1)}{\beta_2(0)}},
\]
whence
\[
0 < \frac{\log \frac{\beta_1(1)}{\beta_1(0)} + \log \frac{\beta_2(1)}{\beta_1(1)}}{\log \frac{\beta_2(1)}{\beta_2(0)}}.
\]
Therefore, by the definition of $f$,
\begin{align}\label{above}
0 < f(1).
\end{align}
By (\ref{between}), (\ref{above}), and the negativity of the slope of $f$,
\[
0 < f(1) < f(0) < 1.
\]

A similar argument shows that, if condition 2 holds, then $0 < f(1) < f(0) < 1$.

Assuming condition 3, we can prove that, if $\beta_2(0) < \beta_1(1) < \beta_2(1)$, then
\begin{align}\label{between1}
0 < f(0) < 1,
\end{align}
using the argument given above. Under the same condition, we have
\[
\beta_2(0) < \beta_1(0),
\]
which implies that
\[
\log \beta_1(1) - \log \beta_1(0) + \log \beta_2(1) - \log \beta_1(1) < \log \beta_2(1) - \log \beta_2(0).
\]
From this inequality, we derive
\[
\frac{\log \frac{\beta_1(1)}{\beta_1(0)} + \log \frac{\beta_2(1)}{\beta_1(1)}}{\log \frac{\beta_2(1)}{\beta_2(0)}} < 1.
\]
Therefore, by the definition of $f$,
\begin{align}\label{above1}
f(1) < 1.
\end{align}
By (\ref{between1}), (\ref{above1}), and the positivity of the slope of $f$,
\[
0 < f(0) < f(1) < 1.
\]

A similar argument shows that, if condition 4 holds, then $0 < f(1) < f(0) < 1$.

Under condition 5 and without loss of generality, assume that $\beta_1 = \mu$ and $\beta_2(0) < 1/2 < \beta_2(1)$, which implies
\[
0 < 1 + \log \beta_2(1) < \log \frac{\beta_2(1)}{\beta_2(0)}.
\]
From the above inequality, we derive
\[
0 < \frac{\log \frac{\beta_2(1)}{1/2}}{\log \frac{\beta_2(1)}{\beta_2(0)}} < 1,
\]
whence, by the definition of $f$,
\[
0 < f(x) < 1,
\]
for all $x \in [0,1]$. \qedhere
\end{proof}

\begin{theorem}\label{only one}
Let $\beta_1$ and $\beta_2$ be probability measures on $\{0,1\}$ that satisfy exactly one of the conditions from Lemma \ref{prob meas}, and let $\alpha_1$ be an arbitrary probability measure on $\{0,1\}$. Then $\alpha_1$ is $(\beta_1,\beta_2)$--equivalent to exactly one unique probability measure $\alpha_2$, which is defined by
\[
\alpha_2(0) = \frac{\alpha_1(0)\log \frac{\beta_1(1)}{\beta_1(0)} + \log \frac{\beta_2(1)}{\beta_1(1)}}{\log \frac{\beta_2(1)}{\beta_2(0)}} \,\,\, \text{ and } \,\,\, \alpha_2(1) = 1 - \alpha_2(0).
\]
\end{theorem}

\begin{proof}
By Lemma \ref{prob meas}, $\alpha_2$ is a valid probability measure. Observe that
\[
\alpha_2(0) = \frac{\alpha_1(0)\log \frac{\beta_1(1)}{\beta_1(0)} + \log \frac{\beta_2(1)}{\beta_1(1)}}{\log \frac{\beta_2(1)}{\beta_2(0)}}
\]
if and only if
\[
\alpha_1(0) \bigg( \log \frac{1}{\beta_1(0)} - \log \frac{1}{\beta_1(1)} \bigg) + \log \frac{1}{\beta_1(1)} = \alpha_2(0) \bigg( \log \frac{1}{\beta_2(0)} - \log \frac{1}{\beta_2(1)} \bigg) + \log \frac{1}{\beta_2(1)}.
\]
The above equality holds if and only if
\[
\alpha_1(0)\log \frac{1}{\beta_1(0)} + \alpha_1(1)\log \frac{1}{\beta_1(1)} = \alpha_2(0)\log \frac{1}{\beta_2(0)} + \alpha_2(1)\log \frac{1}{\beta_2(1)},
\]
which implies that $\alpha_1$ is $(\beta_1,\beta_2)$--equivalent to $\alpha_2$. \qedhere


\end{proof}

The following corollary follows from Theorem \ref{only one} and Lemma \ref{equiv norm}.

\begin{corollary}
Let $\beta_1$, $\beta_2$, $\alpha_1$, and $\alpha_2$ be as defined in Theorem \ref{only one}. For all $(S,T) \in FREQ_{\alpha_1} \times FREQ_{\alpha_2}$, $S$ and $T$ are mutually $\beta$--normalizable.
\end{corollary}
\noindent {\bf Acknowledgments.}
We thank an anonymous reviewer of \cite{CasLut15} for posing the question answered by Corollary \ref{zero not ind}. We also thank anonymous reviewers of this paper for useful comments, especially including Observation \ref{rand marg}.
\bibliography{Master}
\newpage
\end{document}